\documentclass[a4paper,9pt,twocolumn]{IEEEtran}
\usepackage{amsmath}
\usepackage{xcolor}
\usepackage{algorithm,algorithmic}
\usepackage{graphicx}
\usepackage{cite}
\usepackage{hyperref}
\newcommand{\scp}[2]{\left.\left\langle #1\vphantom{#2}\right|#2\right\rangle}
\newcommand{\abs}[1]{|{#1}|}

\newcommand{\norm}[2][]{\left\|#2\right\|_{#1}}

\newcommand{\bigO}{\mathcal{O}}

\renewcommand{\AA}{\mathcal{A}}
\newcommand{\II}{\mathcal{I}}

\newcommand{\RR}{\mathbf{R}}

\newcommand{\doublehookrightarrow}%
{\DOTSB\lhook\joinrel\relbar\!\!\!\!\lhook\joinrel\rightarrow}

\DeclareMathOperator{\Id}{Id}
\DeclareMathOperator{\Sign}{Sign}

\DeclareMathOperator{\rg}{rg}

\newtheorem{lemma}{Lemma}
\newtheorem{corollary}[lemma]{Corollary}
\newtheorem{remark}[lemma]{Remark}

\title{Constructing test instances for Basis Pursuit Denoising}
\author{Dirk A. Lorenz\thanks{Institute for Analysis and Algebra, TU
    Braunschweig, 38092 Braunschweig, Germany, tel: +49-531-3917423
    \url{d.lorenz@tu-braunschweig.de}. Support under DFG grant LO
    1436/2-1 (project ``Sparsity and Compressed Sensing in Inverse
    Problems'') within the Priority Program SPP 1324 ``Extraction of
    quantitative information in complex systems'' and DGF grant LO
    1436/3-1 (project ``Sparse Exact and Approximate Recovery'')}}

\begin{document}
\maketitle

\begin{abstract}
  The number of available algorithms for the so-called Basis Pursuit
  Denoising problem (or the related LASSO-problem) is large and keeps
  growing. Similarly, the number of experiments to evaluate and
  compare these algorithms on different instances is growing.
  
  In this note, we present a method to produce instances with exact
  solutions which is based on a simple observation which is related to
  the so called \emph{source condition} from sparse regularization.
\end{abstract}

\textbf{EDICS:} DSP-RECO, DSP-ALGO

\section{Introduction}
\label{sec:introduction}

``Lately, there has been a lot of fuss about sparse approximation.''
is the beginning of the paper~\cite{tropp2006relax} from 2006 and this
note could have started with the same sentence. Three different
minimization problems have gained much attention. We
follow~\cite{vandenberg2008paretofrontiertbasispursuit} and denote
them as follows: For a matrix $A\in\RR^{k\times n}$ and $b\in\RR^k$
and positive numbers $\sigma$, $\lambda$ and $\tau$ we define the
Basis Pursuit Denoising (\cite{chen1998basispursuit}) with constraint by
\begin{equation}
  \label{eq:BPsigma}
  \tag{$\text{BP}_\sigma$}
  \min_x \norm[1]{x}\ \text{ subject to }\ \norm[2]{Ax-b}\leq \sigma,
\end{equation}
the Basis Pursuit Denoising with penalty (\cite{chen1998basispursuit}) by
\begin{equation}
  \label{eq:QPlambda}
  \tag{$\text{QP}_\lambda$}
  \min_x \tfrac12\norm[2]{Ax-b}^2 + \lambda\norm[1]{x},
\end{equation}
and the LASSO (least absolute shrinkage and selection
operator~\cite{tibshirani1996lasso}) by
\begin{equation}
  \label{eq:LStau}
  \tag{$\text{LS}_\tau$}
  \min_x \norm[2]{Ax-b}\ \text{ subject to } \norm[1]{x}\leq\tau.
\end{equation}
All three problems are related: if we denote with
$x_{\text{QP}}(\lambda)$ a solution of~\eqref{eq:QPlambda}, this also
solves~\eqref{eq:BPsigma} for
$\sigma=\norm[2]{Ax_{\text{QP}}(\lambda)-b}$ and~\eqref{eq:LStau} for
$\tau = \norm[1]{x_{\text{QP}}(\lambda)}$ (see
e.g.~\cite{vandenberg2008paretofrontiertbasispursuit,loris2009l1performance}).
However, this relation is implicit and relies in general on the
knowledge of the solutions.  Hence, it is not totally true that these
problems are equivalent.

One may argue, that~\eqref{eq:BPsigma} is harder than the other
problems since its objective is nonsmooth and shall be minimized over
a complicated convex set (e.g.~projecting on this set is
difficult). Moreover, one may argue, that~\eqref{eq:QPlambda} is
harder than~\eqref{eq:LStau} since the latter has a smooth objective
(to be minimized over a somehow simple convex set) while the first has
a nonsmooth objective. Computational experience with with these
problems lead to the same conclusion.
% Sinn und Herkunft von BPDN, LASSO, Klassifikation der Probleme (wie
% Pareto-Paper), Zusammenhang der verschiedenen Formulierungen,

Recently, minimization problems similar to Basis Pursuit Denoising
have appeared in several contexts, e.g. group sparsity (or joint
sparsity)
\cite{vandenberg2009jointsparserecovery,fornasier2008jointsparsity,mishali2008reduceandboost}
for sparse recovery, nuclear norm minimization for low-rank matrix
recovery~\cite{recht2010nuclearnorm} to name just two.
% Etwas über neuere Entwicklungen, verwandte Probleme (Group Sparsity,
% Nuclear Norm, "Convex Geometry" a la Recht).

% Bemerkung wie "BPDN ist Mutter aller Probleme"

\subsection{Notation}
With $\norm[p]{x}$ we denote the $p$-norm of a vector $x\in\RR^n$,
$A^T$ is the transpose of a matrix $A$, the range of a matrix $A$ is
denoted with $\rg A$ and with $\Sign(x)$ we denote the multivalued
sign, i.e.
\[
y\in\Sign(x) \iff y_i\
\begin{cases}
  = 1 & \text{if }\ x_i>0\\
  = -1 & \text{if }\ x_i<0\\
  \in [-1,1] & \text{if }\ x_i=0
\end{cases}.
\]

\section{Construction of instances with known solution}
\label{sec:constr-probl-with}

In this section we illustrate how instances (i.e.~tuples
$(A,b,\lambda)$) can be generated, such that the solution $x^*$
of~\eqref{eq:QPlambda} is known up to machine precision. This is
achieved by prescribing the solution $x^*$ (and the matrix $A$ and the
value $\lambda$) and computing a corresponding right hand side $b$.

The basis is the following simple observation which has a one-line
proof:
\begin{lemma}
  Let $A\in\RR^{k\times n}$, $\lambda>0$ and $x^*\in\RR^n$ and let
  $w\in\rg A^T$ fulfill $w \in\Sign(x^*)$. Then it holds: If $y$ is a
  solution to $A^Ty = w$ and $b$ is defined by $b = \lambda y + Ax^*$,
  then $x^*$ is a solution of~\eqref{eq:QPlambda}.
\end{lemma}
\begin{proof}
  Simply check
  \begin{align*}
    -A^T(Ax^*-b) & = -A^T(Ax^* - \lambda y - Ax^*)\\
    & = \lambda A^Ty = \lambda w \in \lambda\Sign(x^*).
  \end{align*}
  Hence $x^*$ fulfills the necessary and sufficient condition for
  optimality.
\end{proof}

\begin{remark}
  The existence of the vector $w$ is exactly the \emph{source
    condition} used in sparse regularization of ill-posed
  problems. There one shows that a vector $x^\dagger$ for which such a
  vector $w$ exists can be reconstructed from noisy measurements
  $b^\delta$ with $\norm[2]{Ax^\dagger- b^\delta}\leq \delta$ by
  solving~\eqref{eq:QPlambda} with $b^\delta$ instead of $b$ and
  $\lambda\asymp\delta$ and that one achieves a linear convergence
  rate, i.e.~for the solution $x_\lambda^\delta$ one gets
  $\norm[1]{x_\lambda^\delta-x^\dagger} = \bigO(\delta)$,
  see~\cite{grasmair2008sparseregularization,lorenz2010beyondconvergence,grasmair2011conditionsell1}.
\end{remark}
The following corollary reformulates the above lemma in a way which is
more suitable for an algorithmic reformulation.
\begin{corollary}
  \label{cor:construct_b}
  Let $\{1,\dots,n\}$ be partitioned into sets $\II$, $\AA_+$ and
  $\AA_-$ and let $x^*\in\RR^n$ be any vector such that
  \begin{equation}
    \label{eq:compliance_x}
    \begin{split}
      x_i^*&>0, \quad i\in \AA_+\\
      x_i^*&<0, \quad i\in \AA_-\\
      x_i^*&=0, \quad i\in \II      
    \end{split}
  \end{equation}
  and let $\lambda>0$. Furthermore assume that $y\in\RR^k$ fulfills
  \begin{equation}
    \label{eq:condition_on_y}
    \begin{split}
      (A^Ty)_i&=1, \quad i\in \AA_+\\
      (A^Ty)_i&=-1 , \quad i\in \AA_-\\
      |A^Ty|_i&\leq 1, \quad i\in \II      
    \end{split}
  \end{equation}
  and define $b = \lambda y + Ax^*$. Then $x^*$ is a solution
  of~\eqref{eq:QPlambda}.
\end{corollary}
According to this corollary we can construct an instance
$(A,b,\lambda)$ with known solution $x^*$ as follows:
\begin{enumerate}
\item Specify $A\in\RR^{m\times n}$ and a sign-pattern (given by the
  partition $\AA_+$, $\AA_-$, $\II$).
\item Construct a vector $y\in\RR^m$ which
  fulfills~\eqref{eq:condition_on_y}.
\item Choose any $\lambda>0$ and any $x^*\in\RR^n$ which complies with
  the sign-pattern, i.e.~\eqref{eq:compliance_x} holds.
\item Define $b = \lambda y + Ax$.
\end{enumerate}
The vector $y$ can be constructed by several methods which are outline
in Appendix~\ref{app:algorithm}. These methods have been implemented
in the Matlab package \texttt{L1TestPack} in the function
\verb|construct_bpdn_rhs| \footnote{The package is available at
\url{http://www.tu-braunschweig.de/iaa/personal/lorenz/l1testpack}.}.
One should note that a vector $y$ as in
Corollary~\ref{cor:construct_b} need not to exist. Indeed, for a fixed
matrix $A$ not every sign-pattern of $x^*$ can occur as a minimizer of
any~\eqref{eq:QPlambda}. %View as polytopes. Citation?
\begin{remark}
  For injective $A$ everything is much simpler: Since $A^T$ is
  surjective, we can just choose some $w\in\Sign(x^*)$, solve
  $A^Ty = w$ and set $b = \lambda y+Ax^*$.
\end{remark}

We discuss advantages and disadvantages
of our approach:\\
\textbf{Advantages:}
\begin{itemize}
\item The algorithm is independent of the value of $\lambda$ while the
  performance of solvers for~\eqref{eq:QPlambda} usually deteriorates
  for smaller $\lambda$, see,
  e.g.~\cite{figueiredo2007gradproj,figueiredo2009sparsa} and
  Section~\ref{sec:infl-param-lambda}.
\item The algorithm is independent of the dynamic range of the optimal
  value $x^*$, however, several experiments have recorded that the
  performance of solvers for~\eqref{eq:QPlambda} depends greatly on
  the dynamic range, see, e.g.~\cite{becker2009nesta} and Section~\ref{sec:infl-dynam-range}.
\item For square matrices $A$ with full rank, one immediately get a
  desired vector $y$ by solving $A^Ty = w$ for some vector
  $w\in\Sign(x^*)$. While this setting is unusual, e.g., in compressed
  sensing, one encounters such situations in regularization with
  sparsity constraints,
  see~\cite{daubechies2003iteratethresh,bredies2008harditer,lorenz2008reglp,griesse2008ssnsparsity,denis2009sparseholograms,ramlau2008regproptikhonov}.
\end{itemize}
\textbf{Disadvantages}
\begin{itemize}
\item The construction of $b$ from $x^*$ leads to a specific noise
  model, namely, the noise is given by $\lambda y$. Hence, there is no
  control about the noise distribution\footnote{However, one observes
    that the noise level $\norm{Ax^*-b} = \lambda \norm{y}$ is
    proportional to $\lambda$ which, again, motivates that one should
    choose $\lambda$ proportional to the noise level.}. This limits the
  use of instances constructed in this way to the comparison of
  solvers for basis pursuit denoising. For other sparse reconstruction
  methods like matching pursuit algorithms they seem to be useless.
\item The algorithm produces one particular element $w\in\Sign(x^*)$
  and it is not clear if this has any additional properties. Usually,
  several $w\in\Sign(x^*)\cap \rg A^T$ exist and probably the
  proposed method favors a particular form of $w$.
\end{itemize}

\section{Illustrative instances}
\label{sec:illustr-inst}

Numerous papers contain comparisons of different solvers for the three
problems \eqref{eq:BPsigma},~\eqref{eq:QPlambda} and \eqref{eq:LStau},
see
e.g.~\cite{vandenberg2008paretofrontiertbasispursuit,figueiredo2007gradproj,hale2008fixedpointcontinuation,yin2008bregmaniterations,beck2009fista,griesse2008ssnsparsity,becker2009nesta,loris2009l1performance}. Hence,
we not aim at yet another comparison of solvers but try to illustrate,
how different features of the measurement matrix and the solution
influence the difficulty of the problem.

From the zoo of available solvers we have chosen four. The choice was
not uniformly at random but to represent four different classes:
fpc~\cite{hale2008fixedpointcontinuation} as a simple tuning of the
basic iterative thresholding algorithm, FISTA \cite{beck2009fista} as
a representative of the ``optimal algorithms'' in the sense of worst
case complexity, GPSR~\cite{figueiredo2007gradproj} as a highly tuned
basic gradient method and YALL1~\cite{yang2009yall1} as a member of
the class of alternating directions methods\footnote{Sources: fpc
  version 2.0 \url{http://www.caam.rice.edu/~optimization/L1/fpc/},
  GPSR version 6.0 \url{http://www.lx.it.pt/~mtf/GPSR/}, YALL1 version
  1.0 \url{http://yall1.blogs.rice.edu/} and an own implementation of
  FISTA.}. All these solvers proceed iteratively and use (basically)
one application of $A$ and one of $A^T$ for each iteration. Hence, the
runtime of these algorithms is mainly related to the number of
iterations. We did not include higher order solvers like
fss~\cite{lee2006featuresignsearch} or
ssn~\cite{griesse2008ssnsparsity} and also did not use any variant of
homotopy approaches~\cite{loris2008l1pack}.

For algorithms we overrode the implemented stopping criteria by
the criterion that the relative error in the reconstruction
\[
R_n = \frac{\norm{x_n-x^*}}{\norm{x^*}}
\]
falls below a given threshold.

 \newcommand{\setup}[7]{
  \begin{description}
  \item[Dimensions:]\mbox{}\\[-\baselineskip]
    \begin{itemize}
    \item $n={#1}$ variables,
    \item $k={#2}$ measurements
    \end{itemize}
  \item[Matrix $A$:]\mbox{}\\
    #3
  \item[Solution $x^*$:]\mbox{}\\
    $s={#4}$ non-zero entries, #5
  \item[$\lambda$:]
    {#6}
  \item[Results:]\mbox{}\\
    #7
  \end{description}
}

\subsection{Influence of the parameter $\lambda$}
\label{sec:infl-param-lambda}

Here we consider a standard example from compressed sensing, namely a
sensing matrix $A$ which consists of random rows of a DCT matrix. The
setup is as follows: \setup{1000}%
{200}%
{Random rows of a DCT matrix}%
{20}%
{magnitude normally distributed with mean zero and variance one.}%
{$10^{-1}$, $10^{-2}$, $10^{-4}$}%
{In general, all solver slow down for smaller values of
  $\lambda$. However, some solvers depend greatly on the size of
  $\lambda$, see Figure~\ref{fig:example1}.}

\begin{figure*}
  \centering
  \begin{tabular}{ccc}
    \includegraphics[page=1]{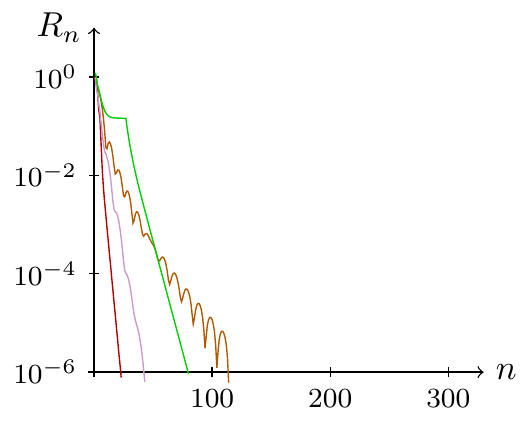}&
    \includegraphics[page=2]{figures_paper/example1}&
    \includegraphics[page=3]{figures_paper/example1}\\
    $\lambda = 10^{-1}$ & 
    $\lambda = 10^{-2}$ & 
    $\lambda = 10^{-4}$  
  \end{tabular}
  \includegraphics{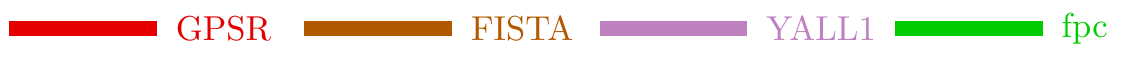}
  \caption{Results for from Section~\ref{sec:infl-param-lambda} on the
    influence of $\lambda$.}
  \label{fig:example1}
\end{figure*}

\subsection{Influence of the sparsity level}
\label{sec:infl-spars-level}

While the construction of a test instance is independent of the
parameter $\lambda$, it gets harder for less sparsity. The behavior of
the solvers with respect to the sparsity level is illustrated by this
example:

\setup{2000}%
{200}%
{Bernoulli ensemble, i.e.~random $\pm 1$}%
{4,\,80}%
{respectively; magnitude normally distributed with mean zero and
  variance one.}%
{$10^{-1}$}%
{Most solvers take longer for less sparsity; however, surprisingly,
  YALL1 is even faster for lower sparsity, see
  Figure~\ref{fig:example2}.}

\begin{figure}
  \centering
  \begin{tabular}{cc}
    \includegraphics[page=1]{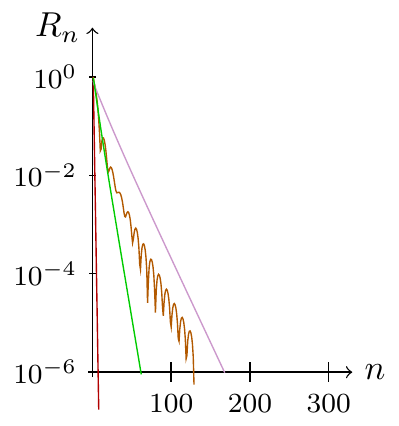}&
    \includegraphics[page=2]{figures_paper/example2}\\
    $s=4$ & 
    $s=80$ 
  \end{tabular}
  \includegraphics{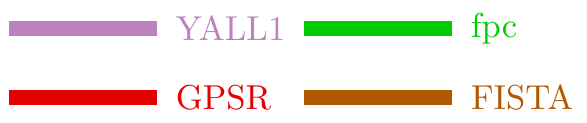}
  \caption{Results for from Section~\ref{sec:infl-spars-level} on the
    influence of the sparsity level $s$.}
  \label{fig:example2}
\end{figure}

\subsection{Influence of the dynamic range of the entries in $x^*$}
\label{sec:infl-dynam-range}

As claimed in the introduction, the \emph{dynamic range}
\[
\Theta(x^*) = \frac{\max\{\abs{x^*}\, :\, x^*\neq 0\}}{\min\{\abs{x^*}\, :\, x^*\neq 0\}}
\]
also influences the performance.

\setup{3000}%
{1000}%
{Union of three orthonormal basis: the identity matrix, the DCT matrix
  and an orthonormalized random matrix}%
{50}%
{with a dynamic range of approximately 9, 701 and 55.000,
  respectively.}%
{$10^{-1}$}%
{Some solvers dramatically slow down for larger dynamic range, see
  Figure~\ref{fig:example3}}

\begin{figure*}
  \centering
  \begin{tabular}{ccc}
    \includegraphics[page=1]{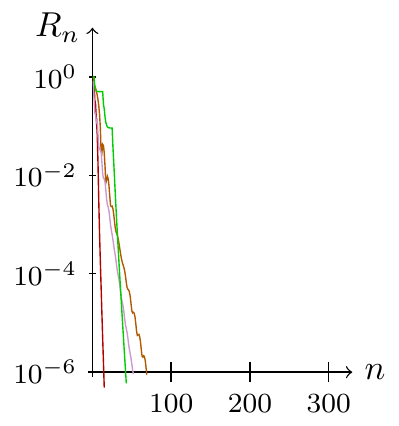}&
    \includegraphics[page=2]{figures_paper/example3}&
    \includegraphics[page=3]{figures_paper/example3}\\
    $\Theta(x^*) \approx 9$ & 
    $\Theta(x^*) \approx 700$ & 
    $\Theta(x^*) \approx 55.000$  
  \end{tabular}
  \includegraphics{figures_paper/legend}
  \caption{Results for from Section~\ref{sec:infl-dynam-range} on the
    influence of the dynamic range.}
  \label{fig:example3}
\end{figure*}

\subsection{Influence of the coherence of $A$}
\label{sec:infl-coher-a}

To illustrate that also a large coherence can cause solvers to slow
down, we have chosen the following setup: We considered square matrices
$A\in\RR^{n\times n}$ which are zero expect on the diagonal and a
certain number $K$ of lower off-diagonals, scaled to have $\norm{A}=1$:
\begin{align*}
  A_K & = c
  \begin{bmatrix}
    1 & 0 & \cdots & \cdots& \cdots& 0\\
    \vdots & \ddots&\ddots & & & \vdots\\
    1 & & \ddots & \ddots & & \vdots\\
    0 & \ddots & &\ddots & \ddots&\vdots\\
    \vdots & \ddots & \ddots&  & \ddots& 0\\
    0 & \cdots & 0 & 1 & \cdots & 1
  \end{bmatrix}.\\[-0.7\baselineskip]
  & \hspace*{20ex}\underbrace{\hspace*{11ex}}_{K\ \text{columns}}
\end{align*}
We also considered the extreme case $K = n$, also known as the
Heaviside matrix. Denoting the columns of $A_K$ by $a_j$, we calculate
the coherence of the matrix $A_K$ as
\[
\mu = \max_{i\neq j}\frac{\scp{a_i}{a_j}}{\norm{a_i}\norm{a_j}} = \sqrt{\frac{K-1}{K}}.
\]

\setup{300}%
{300}%
{Increasingly coherent matrices with $K = 5,\, 40,\, 100,\, 300$}%
{30}%
{Bernoulli, i.e. randomly selected $+1$ and $-1$.}%
{$10^{-1}$}%
{This problem, while with an square and invertible matrix, is known
  the be notoriously hard. Especially for large $K$ all solvers
  deteriorate, see Figure~\ref{fig:example4}.}
\begin{figure*}
  \centering
  \begin{tabular}{cccc}
    \includegraphics[page=1]{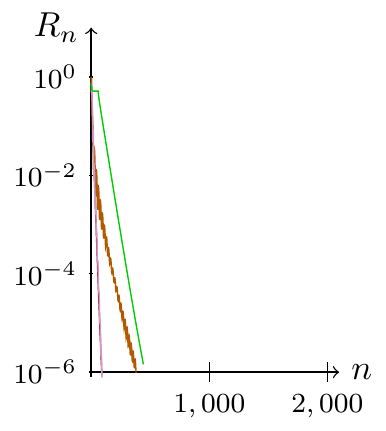}&
    \includegraphics[page=2]{figures_paper/example4}&
    \includegraphics[page=3]{figures_paper/example4}&
    \includegraphics[page=4]{figures_paper/example4}\\
    $K=5$, $\mu=0.89 $& 
    $K=40$, $\mu=0.987 $& 
    $K=100$, $\mu=0.995$ & 
    $K=300$, $\mu=0.998$  
  \end{tabular}
  \includegraphics{figures_paper/legend}
  \caption{Results for from Section~\ref{sec:infl-coher-a} on the
    influence of the coherence.}
  \label{fig:example4}
\end{figure*}

%\subsection{Influence of the condition number of $A$}
%\label{sec:infl-cond-numb}

\appendices
\section{Algorithms}
\label{app:algorithm}
Instead of $y\in\RR^m$ we construct a vector $w\in\RR^n$ such that
\[
w \in \rg A^T \cap \Sign(x^*)
\]
which can be reformulated as
\[
\begin{split}
  w_i&=1, \quad i\in \AA_+\\
  w_i&=-1 , \quad i\in \AA_-\\
  |w_i|&\leq 1, \quad i\in \II      
\end{split}
\]
and $w\in\rg A^T$. Then $y$ can be found by solving $A^Ty = w$.

\subsection{Solution by projection onto convex sets}
\label{sec:solut-proj-onto}

The condition $w\in\rg A^T \cap \Sign(x^*)$ can be seen as a convex
feasibility problem~\cite{bauschke1996convexfeasibility} since both
the sets $\rg A^T$ and $\Sign(x^*)$ are convex. Moreover, the
projection onto each set is computationally feasible: The projection
onto the range of $A^T$ can be calculated explicitly, e.g.~with the
help of QR factorization. If $A^T = QR$ with orthonormal $Q$ and upper
triangular $R$, the projection $P_{\rg A^T}$ is given by $P_{rg A^T} =
Q(:,1:k) Q(:,1:k)^T$. Projecting onto the convex set $\Sign(x^*)$ is
even simpler: Set the fixed components to $\pm 1$ respectively and
clip the others by $x\mapsto\max(\min(x,1)x,-1)$. We done the
projection onto $\Sign(x^*)$ by $P_{\Sign(x^*)}$.

Now we find $w$ by alternatingly project an initial guess onto both
sets, a strategy knows as \emph{projection onto convex sets} (POCS)
\cite{cheney1959pocs,gubin1967pocs}. This is given as pseudo code in
Algorithm~\ref{alg:pocs}.
\begin{algorithm}
  \caption{Calculation of $y$ by POCS}
  \label{alg:pocs}
  \begin{algorithmic}[1]
    \REQUIRE Input $A\in\RR^{m\times n}$, a partition $\AA_+$, $\AA_-$
    and $\II$ of $\{1,\dots,n\}$ (coded as $\Sign(x^*)$), a tolerance
    $\epsilon>0$ and an initial guess $w_0$.
    
    %\STATE Calculate the projection matrix $P_{\rg A^T}$.
    \FOR{$i=0,1,\dots$}
    \STATE $v^n = P_{\rg A^T}w^n$
    \STATE $w^{n+1} = P_{\Sign(x^*)}v^n$
    \IF{$\max(\norm{v^n-w^n},\norm{w^{n+1}-v^n})\leq \epsilon$}
    \STATE break
    \ENDIF
    \ENDFOR
    \STATE Solve $A^Ty = w$
    \RETURN $y$
  \end{algorithmic}
\end{algorithm}

\subsection{Solution by quadratic programming}
\label{sec:solut-quadr-progr}

We sketch another approach by quadratic programming: We call $\AA =
\AA_+\cup\AA_-$ the \emph{active set} and $\II$ the \emph{inactive
  set} and define $s\in\RR^\AA$ by
\begin{equation}
  \label{eq:def_s}
  \begin{split}
    s_i&=1, \quad i\in \AA_+\\
    s_i&=-1 , \quad i\in \AA_-.
  \end{split}
\end{equation}
Furthermore we denote with $P_\AA:\RR^n\to\RR^\AA$ the projection
which deletes the ``inactive'' components and with
$P_\II:\RR^n\to\RR^\II$ the projection which deletes in ``active''
components and the respective adjoint $P_\AA^T$ and $P_\II^T$ which
fill up the vectors be zeros. With this notation, we aim at finding
$w\in\rg A^T$ such that
\[
P_\AA w = s,\ \text{ and }\ \norm[\infty]{P_\II w}\leq 1.
\]
To fulfill the condition $w\in\rg A^T$ we use the orthogonal
projection on $\rg A^T$, denoted by $P_{\rg A^T}$ and require $P_{\rg
  A^T}w = w$. Since $w$ is determined on the active set $\AA$ we
rewrite is as
\begin{equation}
  \label{eq:def_w}
  w = P_\AA^Ts + P_\II^T z
\end{equation}
with a $z\in\RR^\II$. Putting this together we have to find a vector
$z\in\RR^\II$ such that
\[
(P_{\rg A^T} - \Id)P_\II^T z = (\Id - P_{\rg A^T})P_\AA^Ts,\quad
\norm[\infty]{z}\leq 1.
\]
We the abbreviations 
\begin{equation}
  \label{eq:def_barP_barv}
  \begin{split}
    \bar P &= (P_{\rg A^T} - \Id)P_\II^T\\
    \bar v &= (\Id -P_{\rg A^T})P_\AA^Ts
  \end{split}
\end{equation}
we reformulate this as the optimization problem
\begin{equation}
  \label{eq:min_prob_infty}
  \min_{z\in\RR^\II} \tfrac12\norm{\bar Pz-\bar v}^2\ \text{s.t.}\ \norm[\infty]{z}\leq 1.
\end{equation}
This quadratic programming or constrained regression problem can be
solved by various methods~\cite{boyd2004convexoptimization} including
the simple gradient projection~\cite{goldstein1965descent} or the
conditional gradient
method~\cite{frank1956quadprog,beck2004conditionalgradient}.  Note
that we require that the optimal value of~\eqref{eq:min_prob_infty} is
indeed zero.

Algorithm~\ref{alg:calc_y} gives pseudo-code for calculating $y$.
\begin{algorithm}
  \caption{Calculation of $y$ by quadratic programming}
  \label{alg:calc_y}
  \begin{algorithmic}[1]
    \REQUIRE Input $A\in\RR^{m\times n}$ and a partition $\AA_+$,
    $\AA_-$ and $\II$ of $\{1,\dots,n\}$.
    
    \STATE Set $s$ according to~\eqref{eq:def_s}.

    \STATE Calculate the projection matrix $P_{\rg A^T}$ (e.g.~by
    QR-factorization or singular value decomposition) and define $\bar
    P$ and $\bar v$ according to~\eqref{eq:def_barP_barv}.

    \STATE Calculate $z$ as a solution of~\eqref{eq:min_prob_infty}.

    \IF{$\bar P z = \bar v$}

    \STATE calculate $w$ according to \eqref{eq:def_w}

    \ELSE

    \STATE return ``Error: No solution with this sign-pattern''

    \ENDIF
    \STATE Solve $A^Ty = w$
    \RETURN $y$
  \end{algorithmic}
\end{algorithm}

\bibliographystyle{plain}
\bibliography{/home/dirloren/texmf/bibliography/literature}

\end{document}